\documentclass[cis]{ipart}

\RequirePackage[OT1]{fontenc}
\RequirePackage{amsthm,amsmath,amssymb}
\RequirePackage[numbers,square]{natbib} %
\RequirePackage[colorlinks]{hyperref}
\RequirePackage{graphicx}
\RequirePackage{xcolor}

\newtheorem{thm}{Theorem}
\newtheorem{lemma}{Lemma}

\pubyear{0000}
\volume{0}
\issue{0}
\firstpage{1}
\lastpage{9}

\received{\smonth{12} \sday{7}, \syear{2012}}

\begin{document}

\begin{frontmatter}
\title{Rate-distortion functions of non-stationary Markoff chains and their block-independent approximations}

\begin{aug}
\author{\fnms{Mukul} \snm{Agarwal}
%\thanksref{t1,t2}
\ead[label=e1]{magar@alum.mit.edu}}
\address{Department of Electrical and Computer Engineering\\
Boston University\\
\printead{e1}}

%\thankstext{t1}{Some comment.}
%\thankstext{t2}{First supporter of the project.}
%\thankstext{t3}{Second supporter of the project.}
%\thankstext{t4}{Third supporter of the project}
%\runauthor{First Author et al.}

\affiliation{Department of Electrical and Computer Engineering, Boston University}

\end{aug}

\begin{abstract}
It is proved that the limit of the normalized rate-distortion functions of block independent approximations of an irreducible, aperiodic Markoff chain is independent of the initial distribution of the Markoff chain and thus, is also equal to the rate-distortion function of the Markoff chain. 
\end{abstract}

\begin{keyword}[class=AMS]
\kwd[Primary ]{00K00 {\color{cyan} *}}
\kwd{00K01}
\kwd[; secondary ]{00K02 {\color{cyan}  *}}
\end{keyword}

\begin{keyword}
\kwd{\color{red}* }
\end{keyword}
\end{frontmatter}

\addtolength{\parskip}{0.5\baselineskip}
\setlength{\parindent}{0cm}

\section{Introduction}
Consider a random source which evolves on a finite set. It follows from existing  literature, see for example \cite{Shannon} and \cite{GallagerInformationTheory} (Pages 491-500, in particular, see Definition (9.8.3) and Theorem 9.8.3 for achievability), that the limit of the normalized rate-distortion functions of block-independent approximations of a a stationary, ergodic source is equal to the rate-distortion function of the source. Specializing this theorem to irreducible, aperiodic Markoff chains, it follows that the limit of rate-distortion functions of  block-independent approximations of an irreducible, aperiodic Markoff chain which starts in the stationary distribution is equal to the rate-distortion function of this Markoff chain. It is known that the rate-distortion function of an irreducible, aperiodic Markoff chain is independent of its initial distribution (follows from  \cite{Gray}). In this paper, it will be proved that the limit of the normalized rate-distortion functions of  block-independent approximations of an irreducible, aperiodic Markoff chain is independent of its initial distribution. It follows, then, that the rate-distortion function of an irreducible, aperiodic Markoff chain and the limit of the normalized  rate-distortion functions of its block independent approximations are equal and these functions are independent of the initial distribution of the Markoff chain.

Literature on rate-distortion theory is vast. The seminal works are \cite{Shannon} and \cite{ShannonReliable}. 
A work for rate-distortion theory for random processes is \cite{Kolmogorov}. Much of the  classical point-to-point literature on rate distortion theory gets subsumed  under the books \cite{GallagerInformationTheory} and  \cite{Gray}. Another reference is \cite{Berger}. The reader is refered to these three books and  references therein for the literature on rate-distortion theory. In particular, the reader is referred to \cite{Gray} because non-stationary sources are dealt with in great detail in this book, and the concern here is with a non-stationary process, albeit, a non-stationary Markoff chain. For understanding Markoff chains, the reader is referred to \cite{Gnedenko}, \cite{Shiryaev}, and \cite{Feller1}.

\section{Notation and definitions} \label{NotAndDef}

$\mathbb X$ and $\mathbb Y$ denote the source input and source reproduction spaces respectively. Both are assumed to be finite sets. Asume that $\mathbb X = \mathbb Y$. Assume that the cardinality of $\mathbb X$ is greater than or equal to $2$.
$d: \mathbb X \times \mathbb Y \rightarrow [0, \infty )$ is the single-letter distortion measure. Assume that $d(x,x) = 0 \ \forall x \in \mathbb X$ and that $d(x,y) > 0$ if $x \neq y$. 
 Denote 
\begin{align}
D_{\max} \triangleq \max_{x \in \mathbb X, y \in \mathbb Y}d(x,y), 
D_{\min} \triangleq \min _{\{x \in \mathbb X, y \in \mathbb Y | d(x,y) > 0\}} d(x,y)
\end{align}
In what follows, the distortion levels will be assumed to be strictly greater than $0$.
For $x^n \in \mathbb X^n, y^n \in \mathbb Y^n$, the $n$-letter rate-distortion measure is defined additively:
\begin{align}
d^n(x^n, y^n) \triangleq \sum_{i=1}^n d(x^n(i), y^n(i))
\end{align}
where $x^n(i)$ denotes the $i^{th}$ component of $x^n$ and likewise for $y^n$.

Let $X_1, X_2, \ldots$, be a Markoff chain with transition probability matrix $P$, where each $X_i$ is a random-variable on $\mathbb X$.  For $x, x' \in \mathbb X$, $p_{xx'}$ denotes the probability that the Markoff chain is in state $x'$ at time $t+1$ given that it is in state $x$ at time $t$. $p_{xx'}$ is independent of $t$. Assume that the Markoff chain is irreducible, aperiodic. This implies that it has a stationary distribution, henceforth denoted by $\pi$, which will be reserved exclusively for the stationary distribution. In order to specify the Markoff chain completely, we need to specify its initial distribution. If $X_1 \sim \pi'$ denote the Markoff chain $(X_1, X_2, \ldots)$ by $X_{[\pi', P]}$. Recall that $P$ is the transition probability matrix of the Markoff chain. $X_{[\pi',P]}$ will be called the Markoff $X_{[\pi',P]}$ chain. $X^n_{[\pi', P]}$ will denote $(X_1, X_2, \ldots, X_n)$.

The above mentioned assumptions that $\mathbb X = \mathbb Y$, $d(x,x) = 0$ and $d(x,y)> 0$ is $x \neq y$,  that the distortion levels are strictly greater than zero, and that, the Markoff chain is irreducible, aperiodic, will be made throughtout this paper and will not be re-stated.

A rate $R$ source-code is a sequence:

$<e^n, f^n>_1^\infty$, where $e^n: \mathbb X^n \rightarrow \{1, 2, \ldots, 2^{\lfloor nR  \rfloor} \}$ and $f^n: \{1, 2, \ldots, 2^{\lfloor nR  \rfloor} \} \rightarrow \mathbb Y^n$.

 We say that rate $R$ is achievable for source-coding the Markoff $X_{[\pi', P]}$ source within distortion-level $D$ under the expected distortion criterion  if there exists a rate $R$ source code $<e^n, f^n>_1^\infty$ such that 
\begin{align}
\limsup_{n \to \infty} E \left [  \frac{1}{n} d^n(X^n_{\pi'}, f^n(e^n(X^n_{[\pi', P]}))) \right ] \leq D
\end{align}
%(Note that $X^n_{\pi'}$, is the $n$-length source where $X^n_{\pi'}(1) \sim \pi'$ and the evolution of $X^n_{\pi'}(2), X^n_{\pi'}(3), \ldots$ happens according to $P$ -- this is defined above where $T$ was used in the notation instead of $n$; we have used $T$ in the above definition because we will need to talk about the rate-distortion function of the i.i.d. vector source and we would use $T$ there).
%(to be entirely precise, this might be a $\limsup$ or a $\liminf$, but for now, I'll leave it like this).
The infimum of all achievable rates is the rate-distortion function $R^E_{X_{[\pi', P]}}(D)$.

The block-independent approximation (henceforth shortened to BIA) $X^T_{[\pi',P]}$ source is a sequence of random vectors $(S_1, S_2, \ldots, S_n, \ldots)$, where $S_i$ are independent, and $\forall i$, $S_i \sim X^T_{[\pi',P]}$. To simplify notation, we will sometimes denote $(S_1, S_2, \ldots)$ by $S$. $S^n$ will denote $(S_1, S_2, \ldots, S_n)$. Note that BIA $X^T_{[\pi',P]}$ source is an i.i.d. vector source and will also be called the vector i.i.d. $X^T_{[\pi',P]}$ source. Since the BIA $X^T_{[\pi',P]}$ source is an i.i.d vector source, the rate-distortion function for it is defined in exactly the same way as for an i.i.d. source. The details are as follows:
The source input space for the BIA $X^T_{[\pi',P]}$ source is $\mathbb X^T$ and the source reproduction space is $\mathbb Y^T$. Denote these by $\mathbb S$ and $\mathbb T$ respectively. A generic point in $\mathbb S$ is a $T$-length sequence $s$. The $i^{th}$ component of $s$ is denoted by $s(i)$. A generic point in $\mathbb T$ is a $T$-length sequence $t$. The $i^{th}$ component of $t$ is denoted by $t(i)$. The single letter distortion measure is denoted by $d'$ and is defined as $d'(s,t) \triangleq \sum_{j=1}^T d(s(j),t(j))$. For $s^n \in \mathbb S^n$, $t^n \in \mathbb T^n$, the $n$-letter distortion measure $d'^n$ is defined additively: $d'^n(s^n, t^n) \triangleq \sum_{i=1}^n d'(s^n(i),t^n(i))$. Note that $s$ can be thought of as either a scalar in $\mathbb S$ or a $T$ dimensional vector in $\mathbb X^T$. With this identification, $d' = d^T$ and $d'^n$ can be thought of as $d^{nT}$. A rate $R$ source code is a sequence $<e^n, f^n>_1^\infty$, where $e^n: \mathbb S^n \rightarrow \{1, 2, \ldots, 2^{\lfloor nR  \rfloor} \}$ and $f^n: \{1, 2, \ldots, 2^{\lfloor nR  \rfloor} \} \rightarrow \mathbb T^n$. We say that rate $R$ is achievable for source-coding the BIA $X^T_{[\pi',P]}$ source within distortion-level $D$ under the expected distortion criterion  if there exist a sequence of rate $R$ source codes $<e^n, f^n>_1^\infty$ such that 
\begin{align}
\limsup_{n \to \infty} E \left [  \frac{1}{n} d'^n(S^n, f^n(e^n(S^n))) \right ] \leq D
\end{align} 
The infimum of all achievable rates corresponding to a given distortion level $D$ is the operational rate-distortion function at that distortion level, henceforth denoted by 
$R^E_{X^T_{[\pi',P]}}(D)$. The normalized rate-distortion function at block-length $T$  and distortion level $D$ is defined as 
\begin{align}
\frac{1}{T} R^E_{X^T_{[\pi',P]}}(TD)
\end{align}
and the limit is
\begin{align} \label{Normalized}
\lim_{T \to \infty} \frac{1}{T} R^E_{X^T_{[\pi',P]}}(TD)
\end{align}
The theorems in this paper prove the equality of $R^E_{X_{[\pi', P]}}(D)$ and (\ref{Normalized}), and that these functions do not depend on $\pi'$. The statements of these theorems are stated in Section \ref{Lemma}. Before that, we carry out a discussion on the rate-distortion function of a non-stationary Markoff chain.
\subsection{Discussion}
To be entirely correct, the rate-distortion function of a Markoff source should be defined as follows: Let $n$ be the block-length. Denote $U_i \triangleq X_{(i-1)n+1}^{in}$. Each $U_i$ is thus, a random vector of length $n$. Let $<e^n, f^n>_1^\infty$ be a source to code the source $X_{[P, \pi']}$. When the block length is $n$, we would like to use the source-code successively over all intervals of time of block-length $n$. Thus, it is more logical to define the distortion as:
\begin{align}
\limsup_{n \to \infty} \sup_{i \in \mathbb N} E \left [\frac{1}{n} d^n(U_i, f^n(e^n (U_i))) \right ]
\end{align}
and correspondingly define the rate-distortion function. This does not end up making a difference, and hence, we stick to the originally given definition for distortion. Note that if $\pi' = \pi$, the stationary distribution, the $\sup$ in the above definition can be  removed since the distribution of $X_{(i-1)n+1}$ is independent of $i$.
\section{The theorems} \label{Lemma}
\begin{thm}\label{OT1}
$R^E_{X_{[\pi', P]}}(D) = R^E_{X_{[\pi,P]}}(D)$ where $\pi$ is the stationary distribution and $\pi'$ is an arbitrary probability distribution on $\mathbb X$
\end{thm}
\begin{proof}
Follows from \cite{Gray} or see Appendix \ref{EQM} for an independent proof tailored for Markoff chains.
\end{proof}
\begin{thm} \label{OT2}
For $D > 0$,
\begin{align} \label{MainTheorem}
\lim_{T \to \infty} \frac{1}{T} R^E_{X^T_{[\pi', P]}}(TD) \ \mbox{exists, and is independent of $\pi'$} 
\end{align}
\end{thm}
This theorem will be proved in Section \ref{Pf}.
\begin{thm}\label{OT3}
\begin{align}
R^E_{X_{[\pi',P]}}(D) = R^E_{X_{[\pi,P]}}(D) = \lim_{T \to \infty} \frac{1}{T} R^E_{X^T_{[\pi, P]}}(TD) =\lim_{T \to \infty} \frac{1}{T} R^E_{X^T_{[\pi', P]}}(TD)
\end{align}
 where $\pi$ is the stationary distribution and $\pi'$ is an arbitrary distribution on $\mathbb X$.
\end{thm}
\begin{proof}
Follows from Theorems \ref{OT1}, \ref{OT2} and \cite{GallagerInformationTheory}, Pages 490-500.
\end{proof}
In order to prove Theorem \ref{OT2}, we need more notation and this is the subject of the next section. The theorem is proved in the section following the next.
\section{Further notation} \label{FN}
The information-theoretic rate-distortion function of the vector i.i.d. $X^T_{[\pi',P]}$ source is denoted and defined as
\begin{align}
R^I_{X^T_{[\pi',P]}}(D) \triangleq
\inf_{\mathbb W
 }
I(X^T; Y^T)
\end{align}
where $X^T  \sim X^T_{[\pi',P]}$  and  $\mathbb W$  is the set of $W: \mathbb S \rightarrow \mathbb P(\mathbb T)$ defined as
\begin{align}
\mathbb W \triangleq \left \{ W \ \left |\ \sum_{s \in \mathbb S, y \in \mathbb T}
p_{X^T_{[\pi',P]}}(s)W(t|s) d'(s,t) \leq D  \right . \right \}
\end{align} 
where $p_{X^T_{[\pi',P]}}$ denotes the distribution corresponding to $X^T_{[\pi',P]}$. Note that this is the usual definition of the information-theoretic rate-distortion function for an i.i.d. source; just that the source under consideration is vector i.i.d.

Let $s \in \mathbb S$. Denote by $J_{\tau}$ the projection transformation. 
$
J_{\tau}(s) \triangleq (s(\tau+1), s(\tau+2), \ldots s(T))
$.
Fix $s$. Denote
$
\mathbb A \triangleq \{t \in \mathbb S \ | \  J_{\tau}(t) = J_{\tau}(s) \}
$.
Under the distribution induced by $X^T_{[\pi', P]}$, the probability of the set $\mathbb A$ is
\begin{align}
\pi'^{(\tau)}(s(\tau+1))\prod_{i=\tau+1}^{T-1} p_{s(i)s(i+1)}
\end{align}
for some distribution $\pi'^{(\tau)}$ on $\mathbb X$ which satisfies $\pi'^{(\tau)}(x) \to \pi(x)$ as $\tau \to \infty$ $\forall x \in \mathbb X$. Note further, that if $\pi'=\pi$, $\pi'^{(\tau)} = \pi$. For $x \in \mathbb X$, denote  
$
\pi'^{(\tau)}(x) = \pi(x) + {\delta}^{(\tau)}(x)
$
where ${\delta}^{(\tau)}(x) \to 0$ as $\tau \to \infty$. ${\delta}^{(\tau)}(x)$ may be negative.

Denote by $J_{\tau}(X^T_{[\pi', P]})$, the probability distribution on $\mathbb X^{T-\tau}$ which causes the probability of a sequence $r \in \mathbb X^{T-\tau}$ to be
\begin{align}
\pi'^{(\tau)}(r(1)) \prod_{i=1}^{T-1} p_{r(i)r(i+1)}
\end{align}
Note that $J_{\tau}(X^T_{[\pi', P]})$ is the marginal of $X^T_{[\pi', P]}$ on the last $T-\tau$ dimensions. An i.i.d. source can be formed from $J_{\tau}(X^T_{[\pi', P]})$ by taking  a sequence of independent random vectors, each distributed as $J_{\tau}(X^T_{[\pi', P]})$. This will be called the vector i.i.d. $J_{\tau}(X^T_{[\pi', P]})$ source. The rate-distortion function for the vector i.i.d. $J_{\tau}(X^T_{[\pi', P]})$ source in defined in the same way as  the rate-distortion function for the vector i.i.d. $X^T_{[\pi', P]}$ source: 
For $T-\tau$ length sequences, the single-letter distortion measure is defined as $d''(p,q) = \sum_{i=1}^{T-\tau} d(p(i), q(i))$ where $p \in \mathbb X^{T-\tau}$, $q \in \mathbb Y^{T-\tau}$.  The $n$-letter rate-distortion measure is defined additively: $d''^n(p^n,q^n) = \sum_{i=1}^n d''(p^n(i), q^n(i))$ where $p^n \in (\mathbb X^{T-\tau})^n$ and $q^n \in (\mathbb Y^{T-\tau})^n$. A sequence of rate $R$ source codes is a sequence $<e^n, f^n>_1^\infty$, where 
$e^n: (\mathbb X^{T-\tau})^n \rightarrow \{1, 2, \ldots, 2^{\lfloor nR  \rfloor} \}$ 
and  $f^n: \{1, 2, \ldots, 2^{\lfloor nR  \rfloor} \} \rightarrow \mathbb (\mathbb Y^{T-\tau})^n$. The rate-distortion functions for i.i.d. $J_{\tau}(X^T_{[\pi', P]})$ source when the distortion measure is d'' is defined analogously as for the i.i.d. $X^T_{[\pi', P]}$ vector source; the details are omitted. Denote the operational rate-distortion function for the vector i.i.d. $J_{\tau}(X^T_{[\pi'', P]})$ source by $R^E_{J_{\tau}(X^T_{[\pi', P]})}(\cdot)$ and  denote the information-theoretic rate-distortion function for the same source by $R^I_{J_{\tau}(X^T_{[\pi', P]})}(\cdot)$.

For the same reason as that stated before regarding $d'$, $d'' = d^{T-\tau}$ and $d''^n$ can be thought of as $d^{n(T-\tau)}$.

\section{{Proof of the Theorem \ref{OT2}}} \label{Pf}
Before we prove the theorem, note the following:
\begin{lemma}\label{ConvexLemma}
 Let $f: [0, \infty) \rightarrow [0, \infty )$ be a convex $\cup$ non-increasing function. Let $f(0) = K$. Let $0 < a < a'$. Then,
\begin{align}
f(a) - f(a') \leq \frac{K}{a}(a'-a)
\end{align}
\end{lemma}
\begin{proof}
\begin{align}
& f(a)-f(a') \leq \frac{K}{a}(a'-a) \impliedby \frac{f(a)-f(a')}{a'-a} \leq \frac{f(0)-f(a)}{a}  \\
 \impliedby & \frac{f(a')-f(a)}{a'-a} \geq \frac{f(a)-f(0)}{a-0} \impliedby \left (1-\frac{a}{a'} \right )f(0) + \frac{a}{a'} f(a') \geq f(a)  \nonumber \\
 \impliedby & \mbox{Definition of convexity, see for example \cite{V}} \nonumber
\end{align}
\end{proof}

This lemma is a direct result of the the definition of convexity and this observation will be used crucially in the proof of the theorem , which follows below.

Proof of Theorem \ref{OT2}:
\begin{proof}
By the rate-distortion theorem, $R^E_{X^T_{[\pi', P]}}(TD) = R^I_{X^T_{[\pi', P]}}(TD)$.  Comparing definitions with \cite{GallagerInformationTheory}, Page 491,
\begin{align}
\frac{1}{T}R^I_{X^T_{[\pi, P]}}(TD) \ \mbox{ (notation in this document)}  = R_T(D) \ \mbox{(notation in \cite{GallagerInformationTheory})}
\end{align}
By Theorem 9.8.1 in \cite{GallagerInformationTheory}, it follows that
\begin{align}\label{LimExists}
\lim_{T \to \infty} \frac{1}{T} R^E_{X^T_{[\pi, P]}}(TD) \ \mbox{exists}
\end{align}
(\ref{LimExists}) will be used crucially towards the end of the proof.

The proof follows three steps:
\begin{enumerate}
\item
Bound the difference between $R^E_{J_{\tau} (X^T_{[\pi', P]})}(\cdot)$ and $R^E_{J_{\tau} (X^T_{[\pi, P]})}(\cdot)$.
\item
Relate $R^E_{J_{\tau}(X^T_{[\pi', P]})}(\cdot)$ and $R^E_{X^T_{[\pi', P]}}(\cdot)$.
\item
Use these relations to prove the desired result by computing various bounds.
\end{enumerate}

The first step in the proof is to come up with a bound for the difference between $R^E_{J_{\tau} (X^T_{[\pi', P]})}(\cdot)$ and $R^E_{J_{\tau} (X^T_{[\pi, P]})}(\cdot)$. To this end, we first do the same for $R^I_{J_{\tau} (X^T_{[\pi', P]})}(\cdot)$ and $R^I_{J_{\tau} (X^T_{[\pi, P]})}(\cdot)$. To this end, denote the distribution  corresponding to $J_{\tau}(X^T_{[\pi', P]})$ on $\mathbb X^{T-\tau}$ by $Q'$, and  the distribution corresponding to $J_{\tau}(X^T_{[\pi, P]})$ by $Q$. The $l^1$ distance between $Q'$ and $Q$,
\begin{align}\label{l1QQ'}
l^1(Q',Q) & \triangleq \sum_{x^{t-\tau} \in \mathbb X^{T-\tau}} \left |Q'(x^{T-\tau}) -  Q(x^{T-\tau} )\right |  \\
&= \sum_{x^{t-\tau} \in \mathbb X^{t-\tau}} |{\pi'}^{(\tau)}(x^{t-\tau}(1) - \pi^{(\tau)}(x^{t-\tau}(1)| \prod_{i=1}^{T-\tau-1}p_{x^{t-\tau}(i)x^{t-\tau}(i+1)} \nonumber \\
&=\sum_{x \in \mathbb X} |\delta^{(\tau)}(x)| \nonumber \\
& \triangleq \delta^{(\tau)} \nonumber
\end{align}
In the above calculation, we have used the fact that if $\pi'=\pi$, $\pi'^{(\tau)} = \pi$.

Condition (Z) stated in \cite{Hari} holds based on the assumptions we have made, Lemma 2 in \cite{Hari} can be applied, and it follows that for $\tau$ sufficiently large (reasoning stated below after a few lines) and any $T > \tau$,
\begin{align} \label{BoundJtaupipi'}
\left |  \frac{1}{T-\tau} R^I_{J_{\tau}(X^T_{[\pi', P]})}((T-\tau)D) - \frac{1}{T-\tau} R^I_{J_{\tau}(X^T_{[\pi, P]})}((T-\tau)D) \right | \leq K \delta^{(\tau)} \log \frac{1}{\delta^{(\tau)}}
\end{align} 
where
\begin{align}
K = \frac{1}{T-\tau} \frac{7d^*}{\tilde d} \log \left (\left |\mathbb X^{T-\tau} \right | \left |\mathbb Y^{T-\tau} \right | \right )
\end{align}
In (\ref{BoundJtaupipi'}) , $\delta^{(\tau)} \log \frac{1}{\delta^{(\tau)}}$ is defined as zero if $\delta^{(\tau)}$ is zero. $|\mathbb X^{T-\tau}|$ and $|\mathbb Y^{T-\tau}|$ denote the cardinalities of the input and output spaces on which the random source ${J_{\tau}(X^T_{[\pi', P]})}$ is defined. $d^*$ is 
defined as
\begin{align}
d^*   \triangleq \max_{x^{T-\tau} \in \mathbb X^{T-\tau}, ^{yT-\tau} \in \mathbb Y^{T-\tau}}
                               d''(x^{T-\tau}, y^{T-\tau} 
       = (T-\tau) D_{\max}                     
\end{align}
and
$\tilde d$ is defined as 
\begin{align}
\tilde{d}  \triangleq 
 \min_{\{ x^{T-\tau} \in \mathbb X^{T-\tau}, y^{T-\tau} \in \mathbb Y^{T-\tau}\ | \ 
        d''(x^{T-\tau}, y^{T-\tau}) > 0\}} d''(x^{T-\tau}, y^{T-\tau}) 
                = (T-\tau) D_{\min} 
\end{align}

It follows that 
\begin{align}
K = 7 \frac{D_{\max}}{D_{\min}} 
         (\log (|\mathbb X|) +  \log (|\mathbb Y|))
\end{align}

Note that $K$ is a constant independent of $T, \tau, D$.

Also, we said above that (\ref{BoundJtaupipi'}) holds for $\tau$ sufficiently large: this is because by Lemma 2 in \cite{Hari}, we need $\tau$ large enough so that
\begin{align} \label{deltabound}
\delta^{(\tau)} \leq 4 \frac{D_{\min}}{D_{\max}}
\end{align}
which is possible considering the fact that $\delta^{(\tau)} \to 0$ as $\tau \to \infty$, and it is for this reason that we require $\tau$ to be sufficiently large.

Note that the bound in (\ref{deltabound}) is independent of $T, \tau$. It then follows from (\ref{BoundJtaupipi'}) and the equality of information-theoretic and operational rate-distortion functions for i.i.d. sources, that for $\tau$ sufficiently large and any $T > \tau$,
\begin{align} \label{Jpipi'relation}
\left |  \frac{1}{T-\tau} R^E_{J_{\tau}(X^T_{[\pi', P]})}((T-\tau)D) - \frac{1}{T-\tau} R^E_{J_{\tau}(X^T_{[\pi, P]})}((T-\tau)D) \right | \leq K \delta^{(\tau)} \log \frac{1}{\delta^{(\tau)}}
\end{align} 
The bound (\ref{Jpipi'relation}) will be used crucially later, towards the end of the proof.

Next step is to relate $R^E_{J_{\tau}(X^T_{[\pi', P]})}(\cdot)$ and $R^E_{X^T_{[\pi', P]}}(\cdot)$.
 We will argue the following:
\begin{align} \label{REJBound1}
R^E_{X^T_{[\pi', P]}}((T-\tau)D + \tau D_{\max}) \leq R^E_{J_{\tau}(X^T_{[\pi', P]})}((T-\tau)D)
\end{align}
and
\begin{align} \label{REJBound2}
R^E_{J_{\tau}(X^T_{[\pi', P]})}(TD) \leq R^E_{X^T_{[\pi', P]}}(TD)
\end{align}
Very rough idea to prove (\ref{REJBound1}) is the following: Given a sequence of  rate $R$ source codes for the vector i.i.d. $J_{\tau}(X^T_{[\pi', P]})$ source, we can use the same sequence of rate $R$ source-codes for the vector  i.i.d. $X^T_{[\pi', P]}$ source by not coding the time-slots which were not projected onto when defining  $J_{\tau}(X^T_{[\pi', P]})$. These banished slots will incur a maximum distortion of $\tau D_{\max}$ per symbol of $X^T_{[\pi', P]}$. (\ref{REJBound1}) follows. See Appendix \ref{App} for precise argument.

Very rough idea to prove (\ref{REJBound2}) is the following: Consider a two-dimensional random vector $(A,B)$ on some space and the i.i.d. source got by taking i.i.d. copies of $(A,B)$. Consider a distortion measure which is additive over the two dimensions. Consider, also, the i.i.d. source formed by taking identical copies of $A$. Then, for a given distortion level, the rate-distortion function of the vector i.i.d. $(A,B)$ source is greater than or equal to the rate-distortion function of the i.i.d. $A$ source. This is stated more rigorously in Appendix \ref{App}. Note that $J_{\tau}(X^T_{[\pi, P]})$ is a projection of $X^T_{[\pi, P]}$ onto certain dimensions and the distortion measure over these dimensions is additive. (\ref{REJBound2}) follows from this.

Next, we get to Step 3. 
Assuming $TD>\tau D_{\max}$, by replacing $D$ in (\ref{REJBound1})  by
\begin{align}
D =  \frac{TD-\tau D_{\max}}{T-\tau} 
\end{align} 
 and by  (\ref{REJBound2}), it follows that
\begin{align} \label{BKBKBound}
 R^E_{J_{\tau}(X^T_{[\pi', P]})}(TD) 
\leq  R^E_{X^T_{[\pi', P]}}(TD) 
 \leq R^E_{J_{\tau}(X^T_{[\pi', P]})} ( TD-\tau D_{\max})
%& \hspace{2cm} \leq \frac{1}{T-\tau}R^E_{J_{\tau}(X^T_{[\pi', P]})} ((T-\tau)(D - \tau/(T-\tau))   (D_{\max} - D)) \nonumber
\end{align}
It follows from (\ref{BKBKBound}) by rearranging, that
\begin{align} \label{BBound} 
 0 \leq  R^E_{X^T_{[\pi', P]}}(TD) - R^E_{J_{\tau}(X^T_{[\pi', P]})}(TD) \leq
 R^E_{J_{\tau}(X^T_{[\pi', P]})} ( TD-\tau D_{\max}) - R^E_{J_{\tau}(X^T_{[\pi', P]})}(TD)
\end{align}
By noting that  $R^E_{J_{\tau}(X^T_{[\pi', P]})}(D)$ is a non-increasing, convex $\cup$ function of $D$ which is upper bounded by $(T-\tau)\log |\mathbb X|$ at $D = 0$, it follows, assuming that $TD > \tau D_{\max}$,  by Lemma \ref{ConvexLemma} that
\begin{align}\label{BBBBBound}
R^E_{J_{\tau}(X^T_{[\pi', P]})} ( TD-\tau D_{\max}) -  R^E_{J_{\tau}(X^T_{[\pi', P]})}(TD)  \leq \tau D_{\max} \log |\mathbb X|
%& \hspace{5cm} \frac{\log |\mathbb X|}{TD-\tau D_{\max}} \tau D_{\max} \to 0 \ \mbox{as} \ T \to \infty \nonumber 
\end{align}
%Assume, henceforth, that $D > 0$. Assume that $\frac{\tau}{T}$ is small enough so that 
%\begin{align}
%({\tau}/(T-\tau)) (D_{\max} - D) \leq D/2
%\end{align}
%Note that 
%\begin{align}
%\frac{1}{T-\tau}R^E_{J_{\tau}(X^T_{[\pi', P]})}((T-\tau)D)
%\end{align}
%is a non-increasing, non-negative convex $\cup$ function of $D$ which is upper bounded by $\log |\mathbb X|$ at $D = 0$.
%It follows by Lemma \ref{ConvexLemma} that 
%\begin{align} \label{BBBBBound}
%& \frac{1}{T-\tau}R^E_{J_{\tau}(X^T_{[\pi', P]})}((T-\tau)(D - ({\tau}/(T-\tau)) (D_{\max} - D))) - \\ 
%&\hspace{5cm}  \frac{1}{T-\tau}R^E_{J_{\tau}(X^T_{[\pi', P]})}((T-\tau)(DT/(T-\tau))) \nonumber \\
%\leq & \frac{2 \log |\mathbb X|}{D} |DT/(T-\tau) - (D - ({\tau}/(T-\tau)) (D_{\max} - D))|
%\to   0 \ \mbox{as} \ T \to \infty \nonumber
%\end{align}
From (\ref{BBBBBound}) and (\ref{BBound}), it follows that 
\begin{align} \label{E1}
\lim_{T \to \infty} \left [  \frac{1}{T}R^E_{X^T_{[\pi', P]}}(TD)  -  \frac{1}{T}R^E_{J_{\tau}(X^T_{[\pi', P]})}(TD) \right ] = 0 
\end{align}
Note further, by noting that $R^E_{X^T_{[\pi', P]}}(TD) \leq T \log |\mathbb X|$, that 
\begin{align}\label{E2}
\lim_{T \to \infty} \left | \frac{1}{T-\tau}R^E_{X^T_{[\pi', P]}}(TD) -  \frac{1}{T}R^E_{X^T_{[\pi', P]}}(TD) \right |
\leq \lim_{T \to \infty} \frac{\tau}{T(T-\tau)} T \log|\mathbb X| = 0
\end{align}
%This follows by noting that $R^E_{X^T_{[\pi', P]}}(TD) \leq T \log |\mathbb X|$ and thus,
%\begin{align}
%\left [ \frac{1}{T-\tau}R^E_{X^T_{[\pi', P]}}(TD) -  \frac{1}{T}R^E_{X^T_{[\pi', P]}}(TD) \right ]
%= \frac{\tau}{T(T+\tau)}R^E_{X^T_{[\pi', P]}}(TD) \\
%\leq  \frac{\tau}{T(T+\tau)} T \log |\mathbb X|
%\to 0 \ \mbox{as} \ T \to \infty \nonumber
%\end{align}
Also, by noting that $R^E_{J_{\tau}(X^T_{[\pi', P]})}(D)$ is a non-increasing, convex $\cup$ function of $D$ which is upper bounded by $(T-\tau)|\mathbb X|$, it follows by use of Lemma \ref{ConvexLemma} that
\begin{align}\label{E3}
& \lim_{T \to \infty} \left | \frac{1}{T-\tau}R^E_{J_{\tau}(X^T_{[\pi', P]})}(TD) - 
 \frac{1}{T-\tau}R^E_{J_{\tau}(X^T_{[\pi', P]})}((T-\tau)D)  \right |   \leq  \\
& \hspace{6cm} \lim_{T \to \infty} \frac{\log|\mathbb X|}{(T-\tau)D}\tau D \to 0 \ \mbox{as} \ T \to \infty \nonumber
\end{align}
%This is because 
%\begin{align}
% \frac{1}{T-\tau}R^E_{J_{\tau}(X^T_{[\pi', P]})}((T-\tau)D)
%\end{align}
%is a non-increasing, non-negative convex $\cup$ function of $D$ upper bounded by $\log |\mathbb X|$ at $D=0$ . It follows by Lemma \ref{ConvexLemma} that
%\begin{align}
 %\left [ \frac{1}{T-\tau}R^E_{J_{\tau}(X^T_{[\pi', P]})}((T-\tau)(DT/(T-\tau))) - 
 %\frac{1}{T-\tau}R^E_{J_{\tau}(X^T_{[\pi', P]})}((T-\tau)D)  \right ] \\
 %\leq   
 %\frac{\log |\mathbb X|}{D} \left ( \frac{DT}{T-\tau} - D \right )
%\to  0 \ \mbox{as} \ T \to \infty \nonumber
%\end{align}
It follows, then, from (\ref{E1}), (\ref{E2}), (\ref{E3}) and by noting that
\begin{align}
\lim_{n \to \infty} a_n + \lim_{n \to \infty} b_n + \lim_{n \to \infty} c_n = \lim_{n \to \infty} (a_n + b_n + c_n)
\end{align} 
if the three limits on the left hand side exist (follows from definitions, see for example \cite{V}), that
\begin{align} \label{Anotherbound}
\lim_{T \to \infty} \left [ \frac{1}{T}R^E_{X^T_{[\pi', P]}}(TD)  - 
                                                 \frac{1}{T-\tau}R^E_{J_{\tau}(X^T_{[\pi', P]})}((T-\tau)D)  \right ] = 0
\end{align} 
%This implies 
%\begin{align} \label{limlim1}
%\lim_{\tau \to \infty} \lim_{T \to \infty} \left [ \frac{1}{T}R^E_{X^T_{[\pi', P]}}(TD)  - 
%                                                 \frac{1}{T-\tau}R^E_{J_{\tau}(X^T_{[\pi', P]})}((T-\tau)D)  \right ] = 0
%\end{align}
From (\ref{Jpipi'relation}) and (\ref{Anotherbound}), it follows by the use of triangle inequality, that for $\tau$ sufficiently large and $T > \tau$,
\begin{align} \label{EEE1}
\left | \frac{1}{T}R^E_{X^T_{[\pi', P]}}(TD)  - 
                                                 \frac{1}{T-\tau}R^E_{J_{\tau}(X^T_{[\pi, P]})}((T-\tau)D) \right |
                                                  \leq K\delta^{(\tau)} \log \frac{1}{\delta^{(\tau)}} + \kappa_{\tau,T}
\end{align}
for some $\kappa_{\tau,T} \to 0$ as $T \to \infty$.

The above equation holds for $\pi' = \pi$ too, that is, for $\tau$ sufficiently large and $T > \tau$,
\begin{align}\label{EEE2}
\left | \frac{1}{T}R^E_{X^T_{[\pi, P]}}(TD)  - 
                                                 \frac{1}{T-\tau}R^E_{J_{\tau}(X^T_{[\pi, P]})}((T-\tau)D) \right |
                                                  \leq K\delta^{(\tau)} \log \frac{1}{\delta^{(\tau)}} + \eta_{\tau,T}
\end{align}
for some $\eta_{\tau,T} \to 0$ as $T \to \infty$.

From (\ref{EEE1}) and (\ref{EEE2}), by use of the triangle inequality, it follows, that for $\tau$ sufficiently large and $T > \tau$,

\begin{align} \label{KKK1}
\left |  \frac{1}{T}R^E_{X^T_{[\pi', P]}}(TD)  - \frac{1}{T}R^E_{X^T_{[\pi, P]}}(TD) \right |
\leq 2K \delta^{(\tau)} \log \frac{1}{\delta^{(\tau)}}  + \eta_{\tau,T} + \kappa_{\tau,T}
\end{align}

From (\ref{KKK1}) and (\ref{LimExists}),  and by noting that $\delta^{(\tau)} \log \frac{1}{\delta^{(\tau)}} \to 0$ as $\tau \to \infty$, $\eta_{\tau,T}  \to 0$ as $T \to \infty$, and $\kappa_{\tau,T} \to 0$ as $T \to \infty$, it follows that
\begin{align}\label{KKK2}
 \lim_{T \to \infty}  \frac{1}{T}R^E_{X^T_{[\pi', P]}}(TD) \ \mbox{exists and is independent of $\pi'$}
\end{align}
This finishes the proof 
\end{proof}
The assumptions $\mathbb X=\mathbb Y$, $d(x,x) =0$, $d(x,y) > 0$ if $x \neq y$ which have been made are not necessary, and can be replaced by weaker assumptions. Nothing is lost in terms of idea of the proof by making these assumptions, and making these assumptions  prevents one from thinking of pathological cases; for these reasons they have been made.

\section{$\psi$-mixing sources or a variant?}

A set of sources to which this result may be generalizable with the proof technique used is $\psi$-mixing sources or close variants, appropriately defined. See \cite{Prohorov}, \cite{Bradley} and \cite{paperTogether2} for mixing of sources and \cite{Bradley}, \cite{paperTogether2}, in particular, for results on $\psi$-mixing sources. The main property (among others) that made $\psi$-mixing sources amenable to the result in \cite{paperTogether2} is the decomposition in Lemma 1 in \cite{paperTogether2}, wherein, a stationary $\psi$-mixing source is written as a convex combination of an i.i.d. distribution and another general distribution where the i.i.d. distribution dominates as memory is lost with time. Precisely, the equation is Equation (19) in \cite{paperTogether2}:
\begin{align} \label{EqPsiConvex}
\Pr(X_{t+\tau+1}^{t+\tau+T} \in \mathbb B|X_1^t \in \mathbb A) = (1-\lambda_\tau)P_T(\mathbb B) + \lambda_{\tau} P'_{t, \tau, T, \mathbb A}(\mathbb B)
\end{align}
where $\lambda_{\tau} \to 0$ as $\tau \to \infty$. This lemma, though, required stationarity. If a variant of (\ref{EqPsiConvex}) would hold for non-stationary sources, then, there is a possibility that the result in this paper be generalized to such sources. Irreducible, aperiodic Markoff chains statisfy this property, with $P_T(\mathbb B)$ taken as the stationary distribution, and $P'$ is some distribution depending on the initial distribution of the Markoff chain.  An important bound in proving Theorem \ref{OT2} in this paper is the $l^1$ distance between $Q$ and $Q'$, see (\ref{l1QQ'}). This result will hold for sources which satisfy (\ref{EqPsiConvex}) or a variant. Similarly, proving  (\ref{REJBound1}) and (\ref{REJBound2}) in the proof of Theorem \ref{OT2}  or similar equations  may also be possible. The rest of the proof of Theorem \ref{OT2}  is bounding various differences of `close by' rate-distortion functions and this may be possible too. This is just an idea at this point and needs to be studied carefully to see if any of this is at all possible.

\section{Recapitulation and research directions}
In this paper, it was proved that the limit of the normalized rate-distortion functions of block independent approximations of  an irreducible, aperiodic Markoff chain is independent of the initial distribution and is equal to the rate-distortion function of the Markoff chain.

It would be worthwhile trying to generalize this theorem to ergodic sources to the extent possible, not necessarily Markoff sources, in particular, to $\psi$-mixing sources; this would not only make the result general, but also shed light on the `internal workings' of rate-distortion theory. Further, it would be worthwhile trying to prove this result using existing literature, in particular, see if it follows directly from some result, for example, in \cite{Gray}; this would help with generalization and insight into the `internal workings' of rate-distortion theory, too.

\section{Acknowledgements}
The author is infinitely grateful to Prof. Robert Gray for ten's of hours of his time spent in insightful discussions with the author.
\bibliographystyle{IEEEtran}
\bibliography{togetherpaperBibliography}

\appendix

\section{Proof of Theorem \ref{OT1}} \label{EQM}
\begin{proof}
Consider two Markoff chains $X_{[\pi', P]} = (X'_1, X'_2, \ldots)$ and $X_{[\pi'', P]} = (X''_1, X''_2, \ldots)$, where $\pi'$ and $\pi''$ probability distributions on $\mathbb X$. 

Denote $(X'_1, X'_2, \ldots, X'_n)$ by ${X'}^n$ and $(X''_1, X''_2, \ldots, X''_n)$ by ${X''}^n$.  Let $\tau$ be an integer. Think of $n$ as large and $\tau$ to be much smaller than $n$. Denote the distribution of $X'_{\tau}$ by $\mu'$ and the distribution of $X''_{\tau}$ by $\mu''$.  Denote, $\epsilon' = \sum_{x \in \mathbb X} |\mu'(x) - \pi(x)|$ and $\epsilon'' = \sum_{x \in \mathbb X} |\mu''(x) - \pi(x)|$ , where $\pi$ is the stationary distribution of the Markoff chain (note that both Markoff chains have the same transition probability matrix $P$). For every $\epsilon > 0$, $\exists \tau^*_{\epsilon}$ such that $\epsilon'<\frac{\epsilon}{2}$ and $\epsilon'' < \frac{\epsilon}{2} \forall \tau \geq \tau^*_{\epsilon}$. Let $<e^n, f^n>_1^\infty$ be a source-code. Let the block-length be $n$. Think of $n$ large and $\tau^*_{\epsilon} << n$. Use $(e^n, f^n)$ to code ${K'}^n \triangleq (X'_{\tau^*_{\epsilon}+1}, X'_{\tau^*_{\epsilon} + 2}, \ldots, X'_{\tau^*_{\epsilon}+n})$ and ${K''}^n \triangleq (X''_{\tau^*_{\epsilon}+1}, X''_{\tau^*_{\epsilon} + 2}, \ldots, X''_{\tau^*_{\epsilon}+n})$. 
Note that 
\begin{align} \label{O46}
\left | E \left [  \frac{1}{n} d^n({K''}^n, f^n(e^n({K'}^n)))  \right ] - 
 E \left [  \frac{1}{n} d^n({K'}^n, f^n(e^n({K'}^n))) \right ]  \right |  \leq \epsilon D_{\max}
\end{align}
For $\delta > 0$ (think of $\delta$ small), $D > 0$,  let $<{e'}^n, {f'}^n>_1^\infty$ be a source-code with rate $\leq R^E_{[P, \pi']}(D) + \delta$ to code the $X_{[P, \pi']}$ source with distortion $D$. Construct a source code $<{e''}^n, {f''}^n>_1^\infty$ to code the $X_{[P, \pi'']}$ source as follows. When the block-length is $n+\tau^*_{\epsilon}$, code $X''_1, X''_2, \ldots, X''_{\tau^*_{\epsilon}}$ arbitrarily. and code $(X''_{\tau^*_{\epsilon}+1}, X''_{\tau^*_{\epsilon}+2}, \ldots, X''_{\tau^*_{\epsilon}+n})$ using $({e'}^n, {f'}^n)$. It follows, by calculation of the distortion achieved for $(X''_1, X''_2, \ldots, X''_{\tau^*+n})$ by use of this code, in the process, using (\ref{O46})
\begin{align}
R^E_{X_{[\pi'', P]}} \left (D+ \frac{\tau^*_{\epsilon}}{\tau^*_{\epsilon}+n} D_{\max} + \epsilon D_{\max} \right ) \leq R^E_{X_{[\pi', P]}}(D) + \delta
\end{align}
$\epsilon$ can be made arbitrarily small, $\tau^*_{\epsilon}$ will depend on $\epsilon$ and $n$ can be made arbitrarily large. It follows that for every $\alpha > 0$, every $\delta > 0$,
$
R^E_{X_{[\pi'', P]}}(D + \alpha ) \leq R^E_{X_{[\pi', P]}}(D) + \delta
$.
By the continuity of $R^E_{X_{[\pi', P]}}(D)$ in $D$, it follows that $\forall \delta > 0$
$
R^E_{X_{[\pi'', P]}}(D) \leq R^E_{X_{[\pi', P]}}(D) + \delta
$. It follows, then, since $\alpha > 0$ can be arbitrarily small, that
$
R^E_{X_{[\pi'', P]}}(D) \leq R^E_{X_{[\pi', P]}}(D)
$.
By interchanging $\pi'$ and $\pi''$, it follows that,
$
R^E_{X_{[\pi', P]}}(D) \leq R^E_{X_{[\pi'', P]}}(D)
$
Thus,
$
R^E_{X_{[\pi', P]}}(D) =  R^E_{X_{[\pi'', P]}}(D)
$.
\end{proof}

\section{Proofs of (\ref{REJBound1}) and (\ref{REJBound2})} \label{App}

To prove (\ref{REJBound1}): 
\begin{proof}
Let $<e^n, f^n>_1^\infty$ be the source code for the  i.i.d. vector $J_{\tau}(X^T_{[\pi, P]})$ source. 

Note that $e^n: (\mathbb X^{T-\tau})^n \rightarrow \{1, 2, \ldots, 2^{\lfloor nR  \rfloor} \}$  and  

$f^n: \{1, 2, \ldots, 2^{\lfloor nR  \rfloor} \} \rightarrow \mathbb (\mathbb Y^{T-\tau})^n$. 

Let $s^n \in \mathbb S^n$ be a realization of $S^n$, the $n$-blocklength vector i.i.d. $X^T_{[\pi, P]}$ source which needs to be coded. 

$
s^n = (s^n(1), s^n(2), \ldots s^n(n))
$
where each $s^n(i) \in \mathbb S$: 

$
s^n(i) = (s^n(i)(1), s^n(i)(2), \ldots, s^n(i)(T))
$.

Recall the projection operator,
$
J_{\tau}(s^n(i)) = (s^n(i)(\tau+1), \ldots, s^n(i)(T))
$.

Denote
$
J^n_{\tau}(s^n) = (J_{\tau}(s^n(1)),J_{\tau}(s^n(2)),\ldots,J_{\tau}(s^n(n)))
$.

Then, $J^n_{\tau}(s^n)$ is an element of $(\mathbb X^{T-\tau})^n$. Denote $f^n(e^n(J^n_{\tau}(s^n))) = t'^n$. 

Note that 
$
t'^n = (t'^n(1), t'^n(2), \ldots, t'^n(n))
$ where

$
t'^n(i) = (t'^n(i)(1), t'^n(i)(2), \ldots, t'^n(i)(T-\tau))
$.

Fix a random $y \in \mathbb Y$. 
Define the extension transformation,

$
E_{\tau}(t'^n(i)) = (y,y, \ldots, y, t'^n(i)(1), t'^n(i)(2), \ldots, t'^n(i)(T-\tau))
$, 
where the initial $y$'s occur $\tau$ times.

Denote 
$
E^n_{\tau}(t'^n) = (E_{\tau}(t'^n(1)), E_{\tau}(t'^n(2)), \ldots, E_{\tau}(t'^n(n)))
$.

Note that $<e^n \circ J^n_{\tau},  E^n_{\tau}\circ f^n>_1^\infty$ is a rate $R$ source code to code the i.i.d. vector $X^T_{[\pi, P]}$ source and that,
$
d'^n(s^n, E^n_{\tau}(f^n(e^n(J^n_{\tau}(s^n))))) \leq d''^n(s'^n, t'^n) + n \tau D_{\max}
$.
(\ref{REJBound1}) follows.
\end{proof}
To prove (\ref{REJBound2}):
\begin{proof}
Let $(A,B)$ be a random vector on $\mathbb A \times \mathbb B$. Let $(A_1, B_1), (A_2, B_2), \ldots$ be a sequence where $(A_i, B_i)$ are independent of each other and $(A_i, B_i) \sim (A, B)$. This sequence is the vector i.i.d $(A, B)$ source. $\mathbb A \times \mathbb B$ is the source space. Let the source reproduction space be $\mathbb A' \times \mathbb B'$. 
$d_1: \mathbb A \times \mathbb A' \rightarrow [0, \infty)$ is a distortion measure.
$d_2: \mathbb B' \times \mathbb B' \rightarrow [0, \infty)$ is a distortion measure. 
Assume that $\mathbb A, \mathbb A', \mathbb B, \mathbb B'$ are finite sets.
Define:
$d_0((a,b), (a',b')) \triangleq d_1(a,a') + d_2(b,b')$. $d_1^n$ and $d_2^n$, $d_0^n$ are respectively defined additively from $d_1$, $d_2$ and $d_0$. We can then define the rate-distortion functions for the i.i.d. $A$ source and the i.i.d. $(A, B)$ source, denoted, respectively, by $R^E_A(\cdot)$ and $R^E_{(A, B)}(\cdot)$.
Then, 
\begin{align} \label{RDAB}
R^E_A(D) \leq R^E_{(A,B)}(D)
\end{align}
(\ref{RDAB}) is proved as follows: Given a code to code the i.i.d. $(A,B)$ source,   think of $B_i$s as a source of common randomness, and use the obvious variant of the same code for coding  the i.i.d. $A$ source. Since the same code is used, (\ref{RDAB}) follows. Existence of a random code with a certain distortion implies the existence of a deterministic code with the same or lesser distortion. From this, (\ref{RDAB}) follows for deterministic codes.
\end{proof}

\end{document}